\documentclass[12pt,draftcls,journal,onecolumn]{IEEEtran}

\usepackage{amssymb,amsthm, amsmath,latexsym}
\usepackage{graphicx}
\usepackage{mathrsfs}
\usepackage{amsfonts}
\usepackage{amssymb}
\usepackage{longtable}
\usepackage{amsmath}
\usepackage{setspace}
\usepackage{caption}
\usepackage[figuresright]{rotating}
\usepackage[misc]{ifsym}
\usepackage{bbm}
\usepackage{makecell}
\usepackage{arydshln}
\usepackage{supertabular}
\usepackage{booktabs}
\usepackage{color}

\newtheorem{theorem}{Theorem}
\newtheorem{lemma}[theorem]{Lemma}
\newtheorem{remark}[theorem]{Remark}
\newtheorem{proposition}[theorem]{Proposition}

\newtheorem{example}[theorem]{Example}

\newcommand{\ord}{{\mathrm{ord}}}

\newcommand{\gf}{{\mathrm{GF}}}

\newcommand{\F}{{\mathbb{F}}}

\newcommand{\C}{{\mathcal{C}}}

\usepackage{blindtext}

\ifCLASSINFOpdf

\else

\fi

\begin{document}
\title{Solutions to Three Conjectures and an Open Problem on Binary BCH Codes
\thanks{
}
}
\author{Xiaoqiang Wang, Jiawei He$^*$, Boru Yi, Dabin Zheng}

\renewcommand{\thefootnote}{\empty}
\footnotetext{\thanks{*Corresponding author. }
\newline \indent Xiaoqiang Wang, Boru Yi, and Dabin Zheng are with the Hubei Key Laboratory of Applied Mathematics, Faculty of Mathematics and Statistics, Hubei University, Wuhan 430062, China (E-mail:  waxiqq@163.com; yibrru@163.com; dzheng@hubu.edu.cn).
\newline \indent Jiawei He is with the School of Mathematics and Information Science, Nanchang Hangkong University, Nanchang 330036, China (E-mail: hjwywh@mails.ccnu.edu.cn).}

\maketitle

\begin{abstract}

 BCH codes are among the most important classes of cyclic codes and have played a central role in coding theory and its applications.
One of the fundamental problems in the study of BCH codes is to determine their exact minimum distances, which directly govern their error-correcting capability.
Although the BCH bound provides a general lower bound, determining the exact minimum distance is often difficult, and many parameter families remain unresolved. In this paper, we investigate three conjectures and an open problem on binary BCH codes proposed by Chen, Xie, and Ding in \cite{Chen59}. We settle these conjectures on the exact minimum distances of three families of binary BCH codes by constructing codewords attaining the BCH bound. Beyond these conjectures, we further study a more challenging family of codes
and determine its minimum distance for some cases. In addition, we study Open Problem 8.4 in \cite{Chen59}: affirmative answers are obtained for the first two length families, while for the third family a sufficient condition is established and a counterexample shows that the unrestricted assertion does not hold in general.
\end{abstract}



%
\IEEEpeerreviewmaketitle

\section{Introduction}

Let $q$ be a prime power and let $\mathbb{F}_q$ denote the finite field with $q$ elements. An
$[n,k,d]$ linear code $\mathcal{C}$ over $\mathbb{F}_q$ is a $k$-dimensional subspace of
$\mathbb{F}_q^n$ with minimum Hamming distance $d$. A linear code $\mathcal{C}$ is called cyclic if
$
(c_0,c_1,\ldots,c_{n-1})\in \mathcal{C}
$
implies
$
(c_{n-1},c_0,\ldots,c_{n-2})\in \mathcal{C}.
$
By identifying a vector $(c_0,c_1,\ldots,c_{n-1})$ with the polynomial
\[
c_0+c_1x+\cdots+c_{n-1}x^{n-1},
\]
a cyclic code of length $n$ over $\mathbb{F}_q$ can be viewed as an ideal of
$\mathbb{F}_q[x]/\langle x^n-1\rangle$.

Assume that $\gcd(n,q)=1$, and let $\beta$ be a primitive $n$-th root of unity over a suitable
extension field of $\mathbb{F}_q$. For an integer $i$, let $m_i(x)$ denote the minimal polynomial
of $\beta^i$ over $\mathbb{F}_q$. For integers $\delta$ and $b$ with $2\leq \delta\leq n$, the BCH
code $\mathcal{C}_{(q,n,\delta,b)}$ is the cyclic code with generator polynomial
\[
g_{(\delta,b)}(x)
=
\operatorname{lcm}
\bigl(
m_b(x),m_{b+1}(x),\ldots,m_{b+\delta-2}(x)
\bigr).
\]
By the BCH bound,
$d\bigl(\mathcal{C}_{(q,n,\delta,b)}\bigr)\geq \delta.$ When $b=1$, the code is called a narrow-sense BCH code. If the length of the code is \( q^m - 1 \) or \( q^m + 1 \),
then $\mathcal{C}_{(q,n,\delta,b)}$ is called a primitive BCH code, or antiprimitive BCH code, respectively.

BCH codes constitute one of the most important classes of cyclic codes and
have been widely used in reliable data transmission and storage. Binary BCH
codes were introduced independently by Bose and Ray-Chaudhuri and by
Hocquenghem \cite{Bose62,Hocquenghem59}, and were subsequently generalized
to arbitrary finite fields by Gorenstein and Zierler
\cite{Gorenstein61}. Since then, the parameters and structural properties of
BCH codes have been extensively investigated.
Among the various classes of BCH codes, primitive narrow-sense BCH codes
of length $q^m-1$ have received the most attention. Many results on their
dimensions, Bose distances, and minimum distances have been obtained; see,
for example,
\cite{Aly07,Augot94,Charpin90,Ding15,Ding17,Liu17,Lid17,Yue15,Dianwu96}.
BCH codes of other special lengths have also attracted considerable
interest. In particular, antiprimitive BCH codes of length $q^m+1$ have
been studied in \cite{Lid017,Liu17,Li2019,Yan2018}, while BCH codes with
lengths such as
$
\frac{q^m-1}{q-1}
$ and
$\frac{q^m-1}{2}
$
have been investigated in
\cite{Li2017,Ling23,Yan2018,Zhu19}.
These studies show that the arithmetic structure of the code length plays
an important role in determining the cyclotomic cosets, defining sets, and
parameters of BCH codes.

Despite this progress, determining the exact minimum distance of a BCH code remains difficult.
The defining zeros usually provide an effective lower bound through the BCH bound, but proving
that this bound is attained requires the construction of a codeword of the corresponding weight.
Consequently, even when the defining set and dimension of a BCH code are known explicitly, its
exact minimum distance may still remain undetermined.

Recently, Chen, Xie, and Ding \cite{Chen59} studied three families of binary BCH codes with
lengths
\[
n=(2^{2s}+1)(2^s-1),\,\,
n=2^{2s}+2^s+1,\text{ and }
n=\frac{4^s-1}{3},
\]
respectively, where $s\geq 2$ is a positive integer. For these families, several parameters were determined, while the exact minimum
distances of certain narrow-sense BCH codes remained open. In particular, the following three
conjectures were proposed in \cite{Chen59}:
\[
d\bigl(\mathcal{C}_{(2,(2^{2s}+1)(2^s-1),5,1)}\bigr)=5,
\,\,
d\bigl(\mathcal{C}_{(2,2^{2s}+2^s+1,3,1)}\bigr)=3,
\text{ and }
d\left(\mathcal{C}_{\left(2,\frac{4^s-1}{3},5,1\right)}\right)=5.
\]

The main purpose of this paper is to settle these three conjectures. The key point is to construct
codewords whose Hamming weights attain the corresponding BCH lower bounds. Although the three
families have different arithmetic structures, the proofs are based on a common principle: suitable
elements in finite field subgroups are chosen so that the resulting codeword satisfies the prescribed
zero conditions. Depending on the family, this is achieved by using roots of unity, irreducible
polynomials over finite fields, or a lifting argument from primitive BCH codes.
Beyond these conjectures, we further investigate the code $\mathcal{C}_{(2,2^{2s}+2^s+1,5,1)},$ whose minimum distance problem is more involved than that of $\mathcal{C}_{(2,2^{2s}+2^s+1,3,1)}$. Using explicit constructions over finite fields and polynomials, we determine the minimum distance of the code for several infinite families of the parameter $s$.

We further investigate Open Problem~8.4 proposed in \cite{Chen59}, which asks whether there exist
$\delta$ and $b$ such that
\[
\dim\bigl(\mathcal{C}_{(2,n,\delta,b)}\bigr)\geq \frac{n-1}{2}
\,\,
\text{and}\,\,
d\bigl(\mathcal{C}_{(2,n,\delta,b)}\bigr)\geq \frac{\sqrt n}{2}
\] for $n$ being three different values.
For the length $n=(2^{2s}+1)(2^s-1)$ and $n=2^{2s}+2^s+1$, we give affirmative answers by combining a cyclotomic coset
estimate with the BCH bound. For $n=(2^s-1)/\lambda$, where $\lambda>1$ is a given divisor of $2^s-1$, we establish a sufficient condition for the existence of BCH codes with these parameters. We also provide a counterexample demonstrating that the unrestricted assertion fails for some admissible pair $(s,\lambda)$.

The remainder of this paper is organized as follows. Section~II recalls some basic facts on cyclic
codes, BCH codes, cyclotomic cosets, and minimum-distance bounds. Section~III settles the three
minimum-distance conjectures and further studies the minimum distance of
$\mathcal{C}_{(2,2^{2s}+2^s+1,5,1)}$. Section~IV investigates Open Problem~8.4. Finally,
Section~V concludes the paper.

\section{Preliminaries}
Let $q$ be a prime power and let $\mathbb{F}_q$ denote the finite field with $q$ elements. Unless otherwise stated, all codes considered below are
binary.
\subsection{ Cyclotomic cosets and factorization of \(x^n-1\)}
Let $n$ be a positive integer with $\gcd(n,q)=1$, and put
\[
\mathbb{Z}_n=\{0,1,\ldots,n-1\}.
\]
For $i\in\mathbb{Z}_n$, the $q$-cyclotomic coset modulo $n$
containing $i$ is defined by
\[
C_i=\{i,iq,iq^2,\ldots,iq^{\ell_i-1}\}\pmod n,
\]
where $\ell_i$ is the smallest positive integer such that
$
iq^{\ell_i}\equiv i\pmod n.
$
The integer $\ell_i$ is the cardinality of $C_i$, and the smallest
integer in $C_i$ is called its coset leader. The distinct
$q$-cyclotomic cosets form a partition of $\mathbb{Z}_n$.

Let
$
m=\operatorname{ord}_n(q),
$
the multiplicative order of $q$ modulo $n$. Then
$
|C_i|\, \leq\, m
$
for every $i\in\mathbb{Z}_n$.
Let $\alpha$ be a primitive element of $\mathbb{F}_{q^m}$ and set
$
\beta=\alpha^{(q^m-1)/n}.
$
Then $\beta$ is a primitive $n$-th root of unity. For each
$i\in\mathbb{Z}_n$, the minimal polynomial of $\beta^i$ over
$\mathbb{F}_q$ is
\[
m_i(x)=\prod_{j\in C_i}(x-\beta^j).
\]

Let $\mathcal{C}$ be a cyclic code of length $n$ over $\mathbb{F}_q$ with
generator polynomial $g(x)$. The defining set of $\mathcal{C}$ with respect to
$\beta$ is
\[
T(\mathcal{C})=\{i\in\mathbb{Z}_n:g(\beta^i)=0\}.
\]
Since $T(\mathcal{C})$ is a union of $q$-cyclotomic cosets,
$
\deg g(x)=|T(\mathcal{C})|
$
and consequently
$
\dim(\mathcal{C})=n-|T(\mathcal{C})|.
$
In the binary case, we have the useful relation
$
C_{2i}=C_i.
$
This identity will be used repeatedly in the sequel. The following lemma was proved in \cite{Aly07}.
 \begin{lemma}\cite{Aly07}\label{ewnd}
Let \( n \) be a positive integer such that \( q^{\lfloor m/2 \rfloor} < n \leq q^m - 1 \), where \( m = \operatorname{ord}_n(q) \).
Then the \( q \)-cyclotomic coset \( C_s = \{sq^j \bmod n : 0 \leq j \leq m-1\} \) has cardinality \( m \) for all \( s \) in the range \( 1 \leq s \leq nq^{\lfloor m/2 \rfloor}/(q^m - 1) \). In addition, every \( s \) with \( s \not\equiv 0 \pmod{q} \) in this range is a coset leader.
\end{lemma}

\subsection{BCH codes and the BCH bound}

Assume that $\gcd(n,q)=1$, and let $\beta$ be a primitive
$n$-th root of unity over a suitable extension field of $\mathbb{F}_q$.
For each integer $i$, let $m_i(x)$ denote the minimal polynomial of
$\beta^i$ over $\mathbb{F}_q$.

The following are some elementary facts that will be used in the subsequent proof.

\begin{lemma}\label{lem:rootcriterion}
Let $n$ be a positive integer and let $\beta$
be a primitive $n$-th root of unity. Let \(S\) be a subset of \(\{0,1,\ldots,n-1\}\), and define
\[
c(X)=\sum_{j\in S}X^j\in\mathbb F_2[X].
\]
Then the following hold.
\begin{enumerate}
\item If $c(\beta)=0$, then
$
c(X)\in \mathcal{C}_{(2,n,3,1)}.
$

\item If
$
c(\beta)=c(\beta^3)=0,
$
then
$
c(X)\in \mathcal{C}_{(2,n,5,1)}.
$
\end{enumerate}
\end{lemma}

\begin{lemma}\cite{DingP}
 When \(\delta=3\), the code \(\mathcal C_{(2,2^m-1,3,1)}\) is the binary Hamming code with parameters
  \[
  [2^m-1,\,2^m-1-m,\,3],
  \]
  and its generator polynomial is \(m_1(x)\), where \(m\ge 3\).
  \end{lemma}

\begin{lemma}\cite{DingP}\label{eq:5pri}
When \(\delta=5\), the code \(\mathcal C_{(2,2^m-1,5,1)}\) has parameters
  \[
  [2^m-1,\,2^m-1-2m,\,5],
  \]
  and its generator polynomial is \(m_1(x)m_3(x)\), where \(m\ge 4\).
\end{lemma}

The BCH bound will be used throughout the paper.

\begin{lemma}[BCH bound]\label{lem:BCHbound}
Let $\mathcal{C}$ be a cyclic code of length $n$ over $\mathbb{F}_q$. If the
defining set of $\mathcal{C}$ contains
\[
b,b+1,\ldots,b+\delta-2,
\]
then
$
d(\C)\ge\delta.
$
\end{lemma}

\section{Three minimum distance conjectures for binary BCH codes}

In \cite{Chen59}, the minimum distances of three families of binary BCH codes were not completely determined, and three conjectures were proposed concerning their exact values. In this section, we settle these three conjectures by determining the minimum distances of the corresponding BCH codes.

\subsection{The family $n=(2^{2s}+1)(2^s-1)$}

In this subsection, we settle Conjecture 4.8 in \cite{Chen59}.  We start with the following lemma.

\begin{lemma}\label{lem:1}
Let $q=2^{4s}$ and $n=(2^{2s}+1)(2^s-1)$, where $s$ is a positive integer.
Then there exist five pairwise distinct elements
$
x,y,z,u,v\in \mathbb{F}_{q}
$ such that
\begin{eqnarray*}
\begin{cases}
x+y+z+u+v=0,  \\
x^3+y^3+z^3+u^3+v^3=0,\\
x^n=y^n=z^n=u^n=v^n=1.
\end{cases}
\end{eqnarray*}
\end{lemma}

\begin{proof}
Let
\begin{equation}\label{eq:sa}
H=\{a\in \mathbb F_q^*: a^n=1\}.
\end{equation}
Since
$
n=(2^{2s}+1)(2^s-1)
$
divides
\[
2^{4s}-1=(2^{2s}-1)(2^{2s}+1)
=(2^s-1)(2^s+1)(2^{2s}+1),
\]
the subgroup \(H\) is the unique subgroup of \(\mathbb F_q^*\) of order \(n\).
We distinguish two cases according to the residue class of \(s\) modulo \(4\).

\noindent
\textbf{Case 1:} $s\not\equiv 2\pmod 4$.
We first show that \(5\mid n\).
The multiplicative order of \(2\) modulo \(5\) is \(4\).
If \(s\equiv 0\pmod 4\), then
$
2^s\equiv 1\pmod 5,
$
so \(5\mid 2^s-1\).
If \(s\equiv 1\) or \(3\pmod 4\), then \(s\) is odd and
$
2^{2s}\equiv -1\pmod 5,
$
so \(5\mid 2^{2s}+1\).
Thus in all cases with \(s\not\equiv 2\pmod 4\), we have \(5\mid n\), and therefore \(H\) contains an element \(\omega\) of order \(5\).

Take
\[
x=1,\quad y=\omega,\quad z=\omega^2,\quad u=\omega^3,\quad v=\omega^4.
\]
These five elements are pairwise distinct, and since their orders divide \(n\), we have
\[
x^n=y^n=z^n=u^n=v^n=1.
\]
Moreover, because \(1+\omega+\omega^2+\omega^3+\omega^4=0\), we have
\[
x+y+z+u+v=0.
\]
Finally, multiplication by \(3\) permutes the nonzero residues modulo \(5\), so
\[
x^3+y^3+z^3+u^3+v^3
=1+\omega^3+\omega^6+\omega^9+\omega^{12}
=1+\omega+\omega^2+\omega^3+\omega^4
=0.
\]

\medskip

\noindent
\textbf{Case 2:} $s\equiv2\pmod4$. Write \(s=4m+2\) for some integer \(m\ge 0\).
Since \(s\) is even, we have
$
2^s\equiv 1\pmod 3,
$
so \(3\mid 2^s-1\).
Also, because \(2^4\equiv -1\pmod{17}\) and \(2s=8m+4\equiv 4\pmod 8\), we obtain
\[
2^{2s}\equiv 2^4\equiv -1\pmod{17},
\]
so \(17\mid 2^{2s}+1\).
Consequently,
$
51=3\cdot 17 \mid n.
$
Thus \(H\) contains an element of order \(51\), where $H$ is given in (\ref{eq:sa}).
Consider the polynomial
\[
h(X)=X^8+X^7+X^6+X^5+X^4+X+1 \in \mathbb F_2[X].
\]
Let \(\alpha\) be any root of \(h(X)\). Since
$
h(X)\mid X^{51}-1,
$
we have \(\alpha^{51}=1\) in the algebraic closure of \(\mathbb F_2\). Hence, we obtain
$
\operatorname{ord}(\alpha)\mid 51.
$
As \(51=3\cdot 17\), the positive divisors of \(51\) are
$
1,\;3,\;17,\;51.
$
Thus
\[
\operatorname{ord}(\alpha)\in\{1,3,17,51\}.
\]
To prove \(\operatorname{ord}(\alpha)=51\), it remains to exclude the cases \(1,3,17\).
If \(\operatorname{ord}(\alpha)=1\), then \(\alpha=1\), so \(X-1\mid h(X)\). But \(h(1)\neq 0\), hence \(X-1\nmid h(X)\), a contradiction.
If \(\operatorname{ord}(\alpha)=3\), then \(\alpha\neq 1\) and \(\alpha^3=1\). Thus \(\alpha\) is a root of \(X^3-1\) but not of \(X-1\), so \(\alpha\) is a root of
$
X^2+X+1.
$
Since
$
X^3-1=(X+1)(X^2+X+1),
$
a direct computation gives
\[
\gcd(h(X),X^3-1)=1.
\]
Therefore, no root of \(h(X)\) has order \(3\).
If \(\operatorname{ord}(\alpha)=17\), then \(\alpha\neq 1\) and \(\alpha^{17}=1\), so \(\alpha\) is a root of
$
X^{17}-1.
$

It is clear that
\[
X^{17}-1=(X+1)\Phi_{17}(X),
\]
where
$
\Phi_{17}(X)=X^{16}+X^{15}+\cdots+X+1.
$
Since
$
\operatorname{ord}_{17}(2)=8,
$
the polynomial \(\Phi_{17}(X)\) factors over \(\mathbb F_2\) into two irreducible polynomials of degree \(8\). Hence, every irreducible factor of \(X^{17}-1\) has degree \(1\) or \(8\).
A direct computation gives
\[
X^{17}\not\equiv 1\pmod{h(X)},
\]
Also \(h(1)\neq 0\), so \(X+1\nmid h(X)\). If \(\gcd(h(X),X^{17}-1)\) were nontrivial, it would contain an irreducible factor of \(X^{17}-1\) of degree \(1\) or \(8\). The degree \(1\) case is excluded by \(h(1)\neq 0\). In the degree \(8\) case, such a factor would divide \(h(X)\); since \(\deg h(X)=8\), \(h(X)\) would equal that factor up to a nonzero scalar, and therefore \(h(X)\mid X^{17}-1\), contradicting the computation above. Thus,
\[
\gcd(h(X),X^{17}-1)=1.
\]
Consequently, no root of \(h(X)\) has order \(17\).
Hence, every root of \(h(X)\) has multiplicative order \(51\).
Choose \(\theta\in H\) satisfying \(h(\theta)=0\).
Now set
\[
x=1,\quad y=\theta^3,\quad z=\theta^{21},\quad u=\theta^{33},\quad v=\theta^{44}.
\]
The exponents \(0,3,21,33,44\) are pairwise distinct modulo \(51\), so the five elements are pairwise distinct. Since \(51\mid n\), we have
\[
x^n=y^n=z^n=u^n=v^n=1.
\]
Furthermore, one verifies that
$
h(X)\mid 1+X^3+X^{21}+X^{33}+X^{44}.
$
Evaluating at \(X=\theta\), we get
\[
1+\theta^3+\theta^{21}+\theta^{33}+\theta^{44}=0,
\]
that is,
$
x+y+z+u+v=0.
$
Similarly,
$
h(X)\mid 1+X^9+X^{12}+X^{48}+X^{30}.
$
Since \(\theta^{51}=1\), we have
\[
\begin{aligned}
x^3+y^3+z^3+u^3+v^3
=1+\theta^9+\theta^{63}+\theta^{99}+\theta^{132}=1+\theta^9+\theta^{12}+\theta^{48}+\theta^{30}=0.
\end{aligned}
\]

Therefore, in both cases, there exist five pairwise distinct elements
$x,y,z,u,v\in\mathbb{F}_{q}$ satisfying all the required conditions.
\end{proof}

\begin{theorem}\label{Theoremc1}
Let $n=(2^{2s}+1)(2^s-1).$ Then $\mathcal{C}_{(2,n,5,1)}$ has parameters $[n, n-8s, 5]$ for $s\geq2$.
\end{theorem}

\begin{proof}
The dimension of $\mathcal{C}_{(2,n,5,1)}$ can be obtained from Lemma \ref{ewnd}. We show the minimum distance of $\mathcal{C}_{(2,n,5,1)}$ in the following.

By Lemma \ref{lem:BCHbound}, the minimum distance of the binary BCH code \(\mathcal C_{(2,n,5,1)}\) with designed distance \(5\) satisfies
\begin{equation}\label{eq:928}
d(\mathcal C_{(2,n,5,1))}\ge 5.
\end{equation}
To prove equality, it suffices to exhibit a codeword of weight exactly \(5\).
Lemma \ref{lem:1} provides five pairwise distinct elements
$x,y,z,u,v\in \mathbb F_{2^{4s}}$
such that
\[
x+y+z+u+v=0,\qquad
x^3+y^3+z^3+u^3+v^3=0,
\]
and
\[
x^n=y^n=z^n=u^n=v^n=1,
\]
where \(n=(2^{2s}+1)(2^s-1)\).
Since \(n\mid 2^{4s}-1\), these five elements are all \(n\)-th roots of unity in $\mathbb F_{2^{4s}}$. Hence there exist pairwise distinct integers \(a_1,\dots,a_5\) such that
\[
x=\beta^{a_1},\quad y=\beta^{a_2},\quad z=\beta^{a_3},\quad u=\beta^{a_4},\quad v=\beta^{a_5},
\]
where \(\beta\) is a primitive \(n\)-th root of unity. Consider the polynomial
\[
c(X)=X^{a_1}+X^{a_2}+X^{a_3}+X^{a_4}+X^{a_5},
\]
its coefficients are all equal to \(1\), so its Hamming weight is \(5\).
By the conditions in Lemma \ref{lem:1},
\[
c(\beta)=x+y+z+u+v=0
\]
and
\[
c(\beta^3)=x^3+y^3+z^3+u^3+v^3=0.
\]
Applying Lemma \ref{lem:rootcriterion}, the conditions \(c(\beta)=c(\beta^3)=0\) imply that \(c(X)\) is a codeword of \(\mathcal C_{(2,n,5,1)}\).
Thus, we have constructed a codeword of weight \(5\), which yields the upper bound
\[
d(\mathcal C_{(2,n,5,1)})\le 5.
\]
Combining (\ref{eq:928}), we conclude that
$
d(\mathcal C_{(2,n,5,1)})=5.
$
\end{proof}

\begin{remark}
In \cite{Chen59}, the equality
$
d\bigl(\mathcal{C}_{(2,n,5,1)}\bigr)=5
$
was established for $s$ odd and for $s\equiv 0 \pmod 4$, whereas the case \(s \equiv 2 \pmod{4}\) remained open, as formulated in Conjecture 4.8 of \cite{Chen59}. In contrast, Theorem~\ref{Theoremc1} provides a unified argument that covers all possible cases of $s$, and hence completely confirms Conjecture~4.8 in \cite{Chen59}.
\end{remark}

\begin{example}
Let \( s = 2 \). Then \( \mathcal{C}_{(2,n,5,1)} \) has parameters \([51,35,5]\) and is an optimal binary cyclic code \cite[Appendix A]{DingP1}.
\end{example}

\begin{example}
Let \( s = 3 \). Then \( \mathcal{C}_{(2,n,5,1)} \) has parameters \([455,431,5]\).
\end{example}

\subsection{The family $n=2^{2s}+2^s+1$}

We now settle Conjecture 5.3 of \cite{Chen59}.
Let
$q=2^s$ and
$n=q^2+q+1=\frac{q^3-1}{q-1}.$
The subgroup
$
H=\{x\in\F_{q^3}^*:x^n=1\}
$
is precisely the norm-one subgroup
\[
H=\{x\in\F_{q^3}^*:N_{\F_{q^3}/\F_q}(x)=1\}.
\]
The following elementary lemma is the key point. This result may be known, but we could not find a specific reference, so we reproved it.

\begin{lemma}\label{lem:irreduciblecubic}
For every $q=2^s$ with $s\ge2$, there exists $a\in\F_q$ such that the polynomial

\[
f_a(X)=X^3+aX^2+(a+1)X+1
\]
is irreducible over $\F_q$.
\end{lemma}

\begin{proof}  Since \(\operatorname{char}\mathbb F_q=2\), we first observe that for every \(a\in\mathbb F_q\),
$
f_a(0)=1,
$
and
$
f_a(1)=1+a+(a+1)+1=1.
$
Thus, neither \(0\) nor \(1\) is a root of \(f_a(X)\) for any \(a\in\mathbb F_q\).
For $t\in\mathbb F_q\setminus\{0,1\}$, the equation $f_a(t)=0$ is equivalent to
\[
a=\frac{t^3+t+1}{t^2+t}.
\]
Indeed, since $t^2+t\neq0$ for $t\notin\{0,1\}$, the above expression is well-defined, and solving $f_a(t)=0$ for $a$ yields the displayed formula. Hence, the set of all $a\in\mathbb F_q$ for which $f_a(X)$ has a root in $\mathbb F_q$ is contained in the image of the map
\[
\phi:\mathbb F_q\setminus\{0,1\}\to\mathbb F_q,\qquad
\phi(t)=\frac{t^3+t+1}{t^2+t}.
\]
The domain of $\phi$ has cardinality $q-2$, so its image has at most $q-2$ elements. Therefore, there exists $a\in\mathbb F_q$ outside the image of $\phi$. For such an $a$, the polynomial $f_a(X)$ has no root in $\mathbb F_q$. Since a cubic polynomial over a field is reducible over that field if and only if it has a linear factor, we conclude that $f_a(X)$ is irreducible over $\mathbb F_q$.
\end{proof}

\begin{theorem}\label{thm:conj59}
Let $n=2^{2s}+2^s+1$, where $s\ge2$.
Then $\C_{(2,n,3,1)}$ has parameters $[n,n-3s,3]$.
In particular, Conjecture 5.3 of \cite{Chen59} holds.
\end{theorem}

\begin{proof}
The dimension of $\mathcal{C}_{(2,n,3,1)}$ can be obtained from Lemma \ref{ewnd}. We next determine its minimum distance of $\mathcal{C}_{(2,n,3,1)}$.

Set \(q=2^s\). By Lemma~\ref{lem:irreduciblecubic}, there exists \(a\in\mathbb F_q\) such that the cubic
\[
f_a(X)=X^3+aX^2+(a+1)X+1
\]
is irreducible over \(\mathbb F_q\). Let \(\xi\in\mathbb F_{q^3}\) be a root of \(f_a(X)\). Since \(f_a(X)\) is irreducible of degree \(3\), the conjugates of \(\xi\) over \(\mathbb F_q\) are
$\xi$, $\xi^q$ and $\xi^{q^2}.$
The constant term of \(f_a(X)\) is \(1\), so the norm of \(\xi\) satisfies
\[
N_{\mathbb F_{q^3}/\mathbb F_q}(\xi)
=
\xi^{1+q+q^2}
=
1.
\]
Furthermore, a direct computation gives \(f_a(1)=1\). Hence, in characteristic \(2\),
\[
N_{\mathbb F_{q^3}/\mathbb F_q}(\xi+1)
=
\prod_{i=0}^{2}(1+\xi^{q^i})
=
f_a(1)
=
1.
\]
Since \(n=q^2+q+1\), it follows that
\[
\xi^n=1
\quad\text{and}\quad
(\xi+1)^n=1.
\]

Let \(\beta\) be a primitive \(n\)-th root of unity in \(\mathbb F_{q^3}\), and let \(H=\langle \beta\rangle\). Then \(1,\xi,\xi+1\in H\). Moreover, \(\xi\notin\mathbb F_q\) because \(f_a\) is irreducible, and \(f_a(1)=1\) implies \(\xi\neq1\); also \(\xi\neq0\). Thus the three elements \(1,\xi,\xi+1\) are pairwise distinct. Therefore, there exist distinct integers \(i,j\) modulo \(n\) such that
\[
\xi=\beta^i
\quad\text{and}\quad
\xi+1=\beta^j
\]
with \(i,j\neq0\).
Define
$
c(X)=1+X^i+X^j.
$
Then \(c\) has Hamming weight \(3\), and
\[
c(\beta)=1+\beta^i+\beta^j
=
1+\xi+(\xi+1)
=
0.
\]
By Lemma~\ref{lem:rootcriterion}, we have
$
c(X)\in \mathcal C_{(2,n,3,1)}.
$
Hence,
$
d(\mathcal C_{(2,n,3,1)})\le 3.
$

On the other hand, the BCH bound gives
$
d(\mathcal C_{(2,n,3,1)})\ge 3.
$
Therefore equality holds.
\end{proof}

\begin{remark}
In \cite{Chen59}, the equality
$
d\bigl(\mathcal{C}_{(2,n,3,1)}\bigr)=3
$
was established for $s$ being even, while the case $s$ being odd  was left open in Conjecture~5.3 of \cite{Chen59}. In contrast, Theorem \ref{thm:conj59} provides a unified argument that covers all possible cases of $s$, and hence completely confirms Conjecture~5.3 in \cite{Chen59}.
\end{remark}

\begin{example}
Let \( s = 2 \). Then \(\mathcal{ C}_{(2,n,3,1)} \) has parameters \([21,15,3]\), while the optimal binary cyclic code has parameters \([21,15,4]\) \cite[Appendix A]{DingP1}.
\end{example}

\begin{example}
Let \( s = 3 \). Then \( \mathcal{C}_{(2,n,3,1)} \) has parameters \([73,64,3]\) and is an optimal binary cyclic code \cite[Appendix A]{DingP1}.
\end{example}

The preceding conjecture concerns the code \(\mathcal{C}_{(2,n,3,1)}\), for which determining the minimum distance reduces to constructing a codeword of weight \(3\) satisfying the single zero condition \(c(\beta)=0\). The situation for \(\mathcal{C}_{(2,n,5,1)}\) is more involved. In the following, we consider the minimum distance of \(\mathcal{C}_{(2,n,5,1)}\). The following lemma provides the required construction.

\begin{lemma}\label{lem:V1I}
Let \(s\) be a positive integer such that
\[
s\equiv 2 \pmod 6,
\quad\text{or}\quad
s\equiv 4 \pmod 6,
\quad\text{or}\quad 5\mid s
\,\,\text{and}\,\,
15\nmid s
\]
and set
$
n=2^{2s}+2^s+1.
$
Then there exist five pairwise distinct elements
$
x,y,z,u,v\in \mathbb F_{2^{3s}}
$
such that
\begin{eqnarray*}
\begin{cases}
x+y+z+u+v=0,  \\
x^3+y^3+z^3+u^3+v^3=0,\\
x^n=y^n=z^n=u^n=v^n=1.
\end{cases}
\end{eqnarray*}
\end{lemma}

\begin{proof}
We prove this result from the following two cases.

\noindent
\textbf{Case 1:} $s\equiv 2 \pmod 6$,
or
$s\equiv 4 \pmod 6$.  We first show that \(21\mid n\).
Since \(s\) is even, we have
$
2^s\equiv 1 \pmod 3,
$
and hence
$
2^{2s}\equiv 1 \pmod 3.
$
Therefore,
\[
n=2^{2s}+2^s+1\equiv 1+1+1\equiv 0 \pmod 3,
\]
which means that \(3\mid n\).
Now consider the residue of \(2^s\) modulo \(7\). The multiplicative order of \(2\) modulo \(7\) is \(3\). Since \(s\equiv 2\) or \(4\pmod 6\), we have \(3\nmid s\), and thus \(2^s\not\equiv 1\pmod 7\). Moreover, \((2^s)^3\equiv 1\pmod 7\), so \(2^s\) is a nontrivial cube root of unity modulo \(7\). It follows that
\[
(2^s)^2+2^s+1\equiv 0 \pmod 7.
\]
Hence, \(7\mid n\). It follows that
$
21\mid n.
$

Consider the polynomial
\begin{equation}\label{eqhX}
h(X)=X^6+X^4+X^2+X+1\in \mathbb F_2[X].
\end{equation}
Since \(\deg {h(X)}=6\), if \(h(X)\) were reducible over \(\mathbb F_2\), then it would have an irreducible factor of degree at most \(3\). We therefore exclude irreducible factors of degrees \(1\), \(2\), and \(3\).

First,
$h(0)=1$ and $h(1)=1,$
so \(h(X)\) has no linear factor over \(\mathbb F_2\).
The unique monic irreducible quadratic polynomial over \(\mathbb F_2\) is
$$
f_2(X)=X^2+X+1.
$$
It is clear that
$
X^2\equiv X+1 \pmod {f_2(X)}
$ and $
X^3\equiv1\pmod {f_2(X)}.
$
Then
$
X^4\equiv X \pmod {f_2(X)}
$ and
$X^6\equiv1 \pmod {f_2(X)}.
$
Hence,
$$
h(X)
\equiv
1+X+(X+1)+X+1
=
X+1
\not\equiv0
\pmod{f_2(X)}.
$$
Therefore, \(f_2(X)\nmid h(X)\).
There are exactly two monic irreducible cubic polynomials over \(\mathbb F_2\):
$$
f_{3,1}(X)=X^3+X+1\,\, \text{and}\,\,
f_{3,2}(X)=X^3+X^2+1.
$$
It ia obvious that
$
X^3\equiv X+1 \pmod {f_{3,1}(X)}.
$
Thus,
$
X^4\equiv X^2+X \pmod {f_{3,1}(X)},
$ and $
X^6=(X^3)^2\equiv X^2+1 \pmod {f_{3,1}(X)}.
$
It follows that
$$
\begin{aligned}
h(X)
&\equiv
(X^2+1)+(X^2+X)+X^2+X+1\\
&=X^2
\not\equiv0
\pmod{f_{3,1}(X)}.
\end{aligned}
$$
Hence, \(f_{3,1}(X)\nmid h(X)\).

Similarly, we have
$
X^3\equiv X^2+1 \pmod {f_{3,2}(X)}.
$
Then
$
X^4\equiv X^2+X+1 \pmod {f_{3,2}(X)}
$
and
$
X^6\equiv X^2+X \pmod {f_{3,2}(X)}.
$
Hence,
$$
\begin{aligned}
h(X)
&\equiv
(X^2+X)+(X^2+X+1)+X^2+X+1\\
&=X^2+X
\not\equiv0
\pmod{f_{3,2}(X)}.
\end{aligned}
$$
Thus, \(f_{3,2}(X)\nmid h(X)\). Hence, \(h(X)\) has no irreducible factor of degree \(1\), \(2\), or \(3\). Since \(\deg h(X)=6\), every nontrivial factorization of \(h(X)\) would contain a factor of degree at most \(3\). Therefore, (\ref{eqhX})
is irreducible over \(\mathbb F_2\).

Now let \(\theta\) be a root of \(h(X)\). Since
$
h(\theta)=\theta^6+\theta^4+\theta^2+\theta+1=0,
$
we have
\begin{equation}\label{eq:0831}
\theta^6=\theta^4+\theta^2+\theta+1.
\end{equation}
Multiplying both sides by \(\theta\), we obtain
$
\theta^7=\theta^5+\theta^3+\theta^2+\theta.
$
Squaring this identity in characteristic \(2\) gives
\begin{equation}\label{eq:083101}
\theta^{14}
=
\theta^{10}+\theta^6+\theta^4+\theta^2.
\end{equation}
Combining (\ref{eq:0831}) and (\ref{eq:083101}),
we further obtain
$
\theta^8
=
\theta^3+\theta+1,
$
and hence
$
\theta^{10}
=
\theta^5+\theta^3+\theta^2.
$
Substituting these expressions into the formula for \(\theta^{14}\), we get
$$
\begin{aligned}
\theta^{14}
&=
\bigl(\theta^5+\theta^3+\theta^2\bigr)
+\bigl(\theta^4+\theta^2+\theta+1\bigr)
+\theta^4+\theta^2=
\theta^5+\theta^3+\theta^2+\theta+1=
\theta^7+1.
\end{aligned}
$$
Then
$$
\theta^{21}
=
\theta^{14}\theta^7
=
(\theta^7+1)\theta^7
=
\theta^{14}+\theta^7
=
(\theta^7+1)+\theta^7
=
1.
$$
Thus, the multiplicative order of \(\theta\) divides \(21\).

Since \(h(X)\) is irreducible over \(\mathbb F_2\) and has degree \(6\), the minimal polynomial of \(\theta\) over \(\mathbb F_2\) has degree \(6\). Hence, the multiplicative order of \(\theta\) cannot be \(1\), \(3\), or \(7\). Indeed, if \(\theta^3=1\), then \(\theta\) would lie in \(\mathbb F_{2^2}\), so its minimal polynomial over \(\mathbb F_2\) would have degree at most \(2\). Similarly, if \(\theta^7=1\), then \(\theta\) would lie in \(\mathbb F_{2^3}\), so its minimal polynomial would have degree at most \(3\). Both cases contradict the fact that the minimal polynomial of \(\theta\) has degree \(6\).
Therefore, the multiplicative order of \(\theta\) is exactly
$
21.
$
Since every root of the irreducible polynomial \(h(X)\) is a Frobenius conjugate of \(\theta\), all roots of \(h(X)\) have the same multiplicative order. Hence every root of \(h(X)\) has multiplicative order \(21\).

Set
\[
x=1,\quad
y=\theta,\quad
z=\theta^6,\quad
u=\theta^8,\quad
v=\theta^{18}.
\]
The exponents \(0,1,6,8,18\) are pairwise distinct modulo \(21\), so the five elements \(x,y,z,u,v\) are pairwise distinct. Since \(21\mid n\), we have \[
x^n=y^n=z^n=u^n=v^n=1.
\]
It remains to verify the two additive conditions. A direct computation in \(\mathbb F_2[X]\) gives
\[
h(X)\mid 1+X+X^6+X^8+X^{18}.
\]
Evaluating at \(X=\theta\), we obtain
$
1+\theta+\theta^6+\theta^8+\theta^{18}=0,
$
that is,
\[
x+y+z+u+v=0.
\]
Similarly,
\[
\begin{aligned}
x^3+y^3+z^3+u^3+v^3
=1+\theta^3+\theta^{18}+\theta^{24}+\theta^{54}=1+\theta^3+\theta^{18}+\theta^3+\theta^{12}=1+\theta^{12}+\theta^{18},
\end{aligned}
\]
where we used \(\theta^{21}=1\) and the fact that the underlying field has characteristic \(2\). Another direct computation shows that
$
h(X)\mid 1+X^{12}+X^{18}.
$
Then
$
1+\theta^{12}+\theta^{18}=0,
$
and hence
\[
x^3+y^3+z^3+u^3+v^3=0.
\]

\noindent
\textbf{Case 2:} $5\mid s$ and $15\nmid s$. Write
$
s=5k.
$
Since \(15\nmid s\), we have \(3\nmid k\). Put
$
M=2^{10}+2^5+1=1057.
$
Since
\[
2^{15}-1=(2^5-1)(2^{10}+2^5+1)=31M,
\]
we have
$
M\mid 2^{15}-1.
$
We first show that \(M\mid n\). Since \(3\nmid k\), we have
\[
k\equiv 1 \pmod 3
\quad\text{or}\quad
k\equiv 2 \pmod 3.
\]
If \(k\equiv 1\pmod 3\), then
$
2^{5k}\equiv 2^5\pmod M
$
and
$
2^{10k}\equiv 2^{10}\pmod M.
$
Hence,
\[
n=2^{10k}+2^{5k}+1
\equiv 2^{10}+2^5+1
\equiv 0\pmod M.
\]
If \(k\equiv 2\pmod 3\), then
$
2^{5k}\equiv 2^{10}\pmod M
$
and
$
2^{10k}\equiv 2^{20}\equiv 2^5\pmod M.
$
Thus again
\[
n\equiv 2^5+2^{10}+1\equiv 0\pmod M.
\]
Therefore,
$
M\mid n.
$

Since \(M\mid 2^{15}-1\), there exists an element
$
\theta\in \mathbb F_{2^{15}}^{*}
$
of multiplicative order \(M=1057\). As \(5\mid s\), we have
$
15\mid 3s,
$
and hence
$
\mathbb F_{2^{15}}\subseteq \mathbb F_{2^{3s}}.
$
Thus \(\theta\in\mathbb F_{2^{3s}}\).
Let
\[
m(X)
=
X^{15}+X^{10}+X^9+X^8+X^4+X^3+X^2+X+1
\in\mathbb F_2[X].
\]
We prove that \(m(X)\) is irreducible over \(\mathbb F_2\).
Assume \(m(X)=0\). We will prove that \(X\) has order \(1057\) in \(\mathbb F_{2^{15}}^\times\). Consequently,
\[
X^{15}=X^{10}+X^9+X^8+X^4+X^3+X^2+X+1.
\]
Hence,
\[
X^{16}=X^{11}+X^{10}+X^9+X^5+X^4+X^3+X^2+X.
\]
Over \(\mathbb F_2\), it is easy to check that
\[
X^{32}=X^{22}+X^{20}+X^{18}+X^{10}+X^8+X^6+X^4+X^2.
\]
Successive reduction modulo \(m(X)\) gives
\[
X^{32}\equiv X^{14}+X^9+X^8+X^7+X^3+X^2+X\pmod {m(X)}.
\]
Thus,
\[
X^{33}\equiv X^{15}+X^{10}+X^9+X^8+X^4+X^3+X^2\pmod {m(X)}.
\]
Substituting the reduction of \(X^{15}\) and cancelling equal terms in characteristic \(2\),
\[
X^{33}\equiv X+1 \pmod {m(X)}.
\]
Since \(1056=32\cdot 33\), we have
\[
X^{1056}=(X^{33})^{32}\equiv (X+1)^{32}=X^{32}+1\pmod {m(X)}.
\]
Multiplying by \(X\),
\[
X^{1057}\equiv X^{33}+X\equiv (X+1)+X=1\pmod {m(X)}.
\]
Therefore,
$
m(X)\mid X^{1057}-1.
$

Since
$
2^{15}-1=32767=31\cdot 1057,
$
we have \(1057\mid 2^{15}-1\), then \(X^{1057}-1\mid X^{2^{15}-1}-1\). Hence,
\[
m(X)\mid X^{2^{15}}-X.
\]
The polynomial \(X^{2^{15}}-X\) is squarefree over \(\mathbb F_2\), and its irreducible factors are precisely the monic irreducibles over \(\mathbb F_2\) whose degrees divide \(15\). Therefore, \(m(X)\) is squarefree and every irreducible factor of \(m(X)\) has degree \(d\mid 15\). The positive divisors of \(15\) are \(1,3,5,15\); it remains to exclude \(1,3,5\).
Since
\[
m(0)=1\neq 0\,\, \text{and}\,\, m(1)=9\equiv 1\not\equiv 0\pmod 2,
\]
\(m(X)\) has no linear factor over \(\mathbb F_2\).
A cubic irreducible factor would divide \(X^8-X\). It is easy to check that
\[
m(X)\equiv X^4+X+1\pmod{X^8-X},
\]
and the Euclidean algorithm gives
\[
\gcd(X^8-X,X^4+X+1)=1.
\]
Thus \(m(X)\) has no cubic irreducible factor.
A quintic irreducible factor would divide \(X^{32}-X\). Let
\[
R(X)=X^{14}+X^9+X^8+X^7+X^3+X^2
\]
be the remainder of \(X^{32}+X\) modulo \(m(X)\). Since \(X^{32}-X=X^{32}+X\) in characteristic \(2\),
\[
\gcd(m(X),X^{32}-X)=\gcd(m(X),R(X)).
\]
Moreover,
$
m(X)+XR(X)=X^2+X+1,
$
then
\[
\gcd(m(X),R(X))=\gcd(R(X),X^2+X+1).
\]
Substitution gives \(R(X)\equiv 1\pmod{X^2+X+1}\), we have
$
\gcd(m(X),X^{32}-X)=1.
$
Thus, \(m(X)\) has no quintic irreducible factor.
Consequently, every irreducible factor of \(m(X)\) has degree \(15\). Since \(\deg m(X)=15\), \(m(X)\) is irreducible over \(\mathbb F_2\), and
\[
\mathbb F_2[X]/(m(X))\cong \mathbb F_{2^{15}}.
\]
It remains to determine the order of \(X\) in \(\mathbb F_{2^{15}}^{*}\). Since \(1057=7\cdot 151\) with \(7\) and \(151\) prime, and since \(X^{1057}\equiv 1\pmod {m(X)}\), the order of \(X\) is \(1057\) if and only if
\[
X^7\not\equiv 1\pmod {m(X)}\,\,\,\, \text{and}\,\,\,\, X^{151}\not\equiv 1\pmod {m(X)}.
\]
The first is immediate: \(X^7-1\) has degree \(7\), so it cannot be divisible by the degree-\(15\) polynomial \(m(X)\).
For the second, repeated squaring yields
\[
X^{64}\equiv 1+X^3+X^6+X^9+X^{13}+X^{14}\pmod {m(X)},
\]
\[
X^{128}\equiv 1+X+X^4+X^7+X^{11}+X^{12}\pmod {m(X)},
\]
and, using
\[
X^{23}\equiv X+X^3+X^4+X^5+X^7+X^8+X^{10}+X^{12}+X^{13}\pmod {m(X)},
\]
one obtains
\[
X^{151}=X^{128}X^{23}\equiv X^{14}+X^{10}+X^6+X^3+X \pmod {m(X)}.
\]
This is not the constant polynomial \(1\), so \(X^{151}\not\equiv 1\pmod {m(X)}\). Therefore the order of \(X\) in \(\mathbb F_{2^{15}}^\times\) is exactly \(1057\).
Thus
we may choose \(\theta\) such that \(m(\theta)=0\) and
\(\operatorname{ord}(\theta)=1057\). A direct computation in
\(\mathbb F_2[X]\) gives
$$
m(X)\mid
\left(
1+X^{72}+X^{105}+X^{119}+X^{234}
\right)
\,\,\text{and}\,\,
m(X)\mid
\left(
1+X^{216}+X^{315}+X^{357}+X^{702}
\right).
$$
Therefore,
$$
1+\theta^{72}+\theta^{105}+\theta^{119}+\theta^{234}=0
\,\,\text{and}\,\,
1+\theta^{216}+\theta^{315}+\theta^{357}+\theta^{702}=0.
$$

Now set
\[
x=1,\qquad
y=\theta^{72},\qquad
z=\theta^{105},\qquad
u=\theta^{119},\qquad
v=\theta^{234}.
\]
We have
\[
x+y+z+u+v=0
\,\,\text{and}\,\,
x^3+y^3+z^3+u^3+v^3=0.
\]

Finally, since each of \(x,y,z,u,v\) has multiplicative order
dividing \(M\), and \(M\mid n\), we obtain
\[
x^n=y^n=z^n=u^n=v^n=1.
\]

From the above two cases, we obtain the desired result.
\end{proof}

With an analysis similar to that of Theorem \ref{Theoremc1}, we have the following result.

\begin{theorem}
 Let \(s\ge 3\) be an integer satisfying
\[
s\equiv 2\pmod 6,
\quad\text{or}\quad
s\equiv 4\pmod 6,
\quad\text{or}\quad
5\mid s\ \text{and}\ 15\nmid s.
\]
Set
\(
n=2^{2s}+2^s+1.
\)
Then the binary BCH code \(\mathcal C_{(2,n,5,1)}\) has parameters
\[
[n,\; n-6s,\; 5].
\]
\end{theorem}

\begin{example}
Let \( s = 4 \). Then \( \mathcal{C}_{(2,n,5,1)} \) has parameters \([273,249,5]\).
\end{example}

\begin{remark}
It should be pointed out that
$
d\bigl(\mathcal{C}_{(2,\,n,\,5,\,1)}\bigr)
$
is not equal to \(5\) for all positive integers \(s\). For example, by Magma,
$d\bigl(\mathcal{C}_{(2,73,5,1)}\bigr)=6$ if \(s=3\) and
$d\bigl(\mathcal{C}_{(2,7,5,1)}\bigr)=7$ if $s=1$.
\end{remark}

\subsection{The family $n=(4^s-1)/3$}

In this subsection, we settle Conjecture 6.9 in \cite{Chen59}. We first show the following lemma.

\begin{lemma}\label{lm:Con3}
Let $
n=\frac{4^s-1}{3}
$, where $s\ge 2$ is an integer.
Then there exist five pairwise distinct elements
$
x,y,z,u,v\in \mathbb{F}_{2^{4s}}
$
such that
\begin{eqnarray}\label{eq:oct1027}
\begin{cases}
x+y+z+u+v=0,  \\
x^3+y^3+z^3+u^3+v^3=0,\\
x^n=y^n=z^n=u^n=v^n=1.
\end{cases}
\end{eqnarray}
\end{lemma}

\begin{proof}
It is obvious that
$
n\mid 2^{2s}-1.
$
Then every $n$-th root of unity lies in $\mathbb F_{2^{2s}}^*$, which is a subgroup of $\mathbb F_{2^{4s}}^*$.
We distinguish three cases according to the value of $s$.

\medskip

\noindent
\textbf{Case 1:} $s$ is even.
Since $s$ is even, we have
$
4^s\equiv 1\pmod{15}.
$
It follows that
$
5\mid \frac{4^s-1}{3}=n.
$
Therefore, there exists an element $\omega\in\mathbb{F}_{2^{2s}}^*$
of order $5$. Set
\[
x=1,\qquad
y=\omega,\qquad
z=\omega^2,\qquad
u=\omega^3,\qquad
v=\omega^4.
\]
Clearly, $x,y,z,u,v$ are pairwise distinct. Since $5\mid n$, we have
\[
x^n=y^n=z^n=u^n=v^n=1.
\]
Moreover,
$
1+\omega+\omega^2+\omega^3+\omega^4=0,
$
and hence,
\[
x+y+z+u+v=0.
\]
Since $\gcd(3,5)=1$, multiplication by $3$ induces a permutation
of the residue classes modulo $5$. Therefore,
\[
\begin{aligned}
x^3+y^3+z^3+u^3+v^3
=1+\omega^3+\omega^6+\omega^9+\omega^{12}=1+\omega+\omega^2+\omega^3+\omega^4=0.
\end{aligned}
\]

\medskip

\noindent
\textbf{Case 2:}  $s=3$.
In this case,
$
n=\frac{4^3-1}{3}=21.
$
Let $\theta\in\mathbb F_{2^6}$ be a root of the polynomial
\[
h(X)=X^6+X^4+X^2+X+1\in \mathbb F_2[X].
\]
We now prove that the polynomial
$
h(X)
$
is irreducible over $\mathbb F_2$, and that each of its roots has multiplicative order $21$.
First, since
\[
h(0)=1\neq 0
\,\,\text{and}\,\,
h(1)=1\neq 0,
\]
the polynomial $
h(X)
$ has no root in $\mathbb F_2$, and hence no linear factor.
Clearly, the only irreducible quadratic polynomial over $\mathbb F_2$ is
$
f_2(X)=X^2+X+1.
$
Then we have
$
X^2\equiv X+1 \pmod {f_2(X)},
$
Hence,
\[
X^3\equiv X(X+1)\equiv X^2+X\equiv (X+1)+X\equiv 1 \pmod {f_2(X)}.
\]
Consequently,
\[
X^6\equiv (X^3)^2\equiv 1 \pmod {f_2(X)}
\,\,\,\,\text{and}\,\,\,\,
X^4= X^3\cdot X\equiv X \pmod {f_2(X)}.
\]
Substituting these congruences into $h(X)$, we obtain
\[
h(X)\equiv 1+X+(X+1)+X+1\equiv X+1\not\equiv 0\pmod{f_2(X)}.
\]
Thus, \(f_2(X)\) does not divide \(h(X)\).

It is known that
$
f_{3,1}(X)=X^3+X+1
\,\,\text{and}\,\,
f_{3,2}(X)=X^3+X^2+1
$
are irreducible cubic polynomials over \(\mathbb F_2\).
From \(f_{3,2}(X)\), we have $X^4\equiv X^2+X \pmod{f_{3,1}(X)}$,
\[
X^5\equiv X^2+X+1 \pmod{f_{3,1}(X)}\,\, \text{and}\,\,X^6\equiv X^2+1 \pmod{f_{3,1}(X)}.\]
Substitution into \(h(X)\) gives
\[
h(X)\equiv (X^2+1)+(X^2+X)+X^2+X+1\equiv X^2\not\equiv 0\pmod{f_{3,1}(X)}.
\]
Similarly, from \(f_{3,2}(X)\), we have $X^4\equiv X^2+X+1\pmod{f_{3,2}(X)},$
\[
X^5\equiv X+1\pmod{f_{3,2}(X)}\,\,\text{and}\,\,
X^6\equiv X^2+X\pmod{f_{3,2}(X)}.
\]
Then
\[
h(X)\equiv (X^2+X)+(X^2+X+1)+X^2+X+1\equiv X^2+X\not\equiv 0\pmod{f_{3,2}(X)}.
\]
Hence, \(h(X)\) has no irreducible factor of degree at most \(3\). Since \(\deg h(X)=6\), any nontrivial factorization of \(h(X)\) would necessarily contain a factor of degree at most \(3\). Therefore, \(h(X)\) is irreducible over \(\mathbb F_2\).

Now let \(\theta\) be a root of \(h(X)\). From \(h(\theta)=0\) we have
$
\theta^6=\theta^4+\theta^2+\theta+1.
$
Multiplying by \(\theta\) gives
$
\theta^7=\theta^5+\theta^3+\theta^2+\theta.
$
Squaring this identity in characteristic \(2\), we get
\[
\theta^{14}=(\theta^7)^2
=\theta^{10}+\theta^6+\theta^4+\theta^2.
\]
Next, using the expression for \(\theta^6\), we compute
\[
\theta^8=\theta\cdot\theta^7
=\theta^6+\theta^4+\theta^3+\theta^2
=\theta^3+\theta+1,
\]
and then
$
\theta^9=\theta^4+\theta^2+\theta$ and
$
\theta^{10}=\theta^5+\theta^3+\theta^2.
$
Substituting these into the expression for \(\theta^{14}\), we obtain
\[
\theta^{14}
=(\theta^5+\theta^3+\theta^2)+(\theta^4+\theta^2+\theta+1)+\theta^4+\theta^2
=\theta^5+\theta^3+\theta^2+\theta+1
=\theta^7+1.
\]
Then
\[
\theta^{21}
=\theta^{14}\theta^7
=(\theta^7+1)\theta^7
=\theta^{14}+\theta^7
=(\theta^7+1)+\theta^7
=1.
\]
Thus, the order of \(\theta\) divides \(21\).

Since \(\theta\neq 1\), the order is not \(1\). If \(\theta^3=1\), then \(\theta\) would be a primitive third root of unity, whose minimal polynomial over \(\mathbb F_2\) has degree \(2\), contradicting the fact that \(h(X)\) is an irreducible polynomial of degree \(6\) with \(\theta\) as a root.

 Similarly, if \(\theta^7=1\), then \(\theta\) would be a primitive seventh root of unity, whose minimal polynomial over \(\mathbb F_2\) has degree \(3\), again contradicting \(\deg h(X)=6\). Hence, the order of \(\theta\) divides \(21\) but is not a proper divisor of \(21\). Therefore, its multiplicative order is exactly \(21\).
Choose
\[
x=1,\qquad
y=\theta,\qquad
z=\theta^6,\qquad
u=\theta^8,\qquad
v=\theta^{18}.
\]
The exponents
$
0,1,6,8,18
$
are pairwise distinct modulo $21$, and hence
$x,y,z,u,v$ are pairwise distinct. Since $\theta^{21}=1$, we have
\[
x^{21}=y^{21}=z^{21}=u^{21}=v^{21}=1.
\]

A direct computation in $\mathbb{F}_2[X]$ gives
\[
h(X)\mid
\left(1+X+X^6+X^8+X^{18}\right).
\]
Evaluating at $X=\theta$, we obtain
$
1+\theta+\theta^6+\theta^8+\theta^{18}=0,
$
and then
$
x+y+z+u+v=0.
$
Furthermore,
\[
\begin{aligned}
x^3+y^3+z^3+u^3+v^3
&=1+\theta^3+\theta^{18}+\theta^{24}+\theta^{54}\\
&=1+\theta^3+\theta^{18}+\theta^3+\theta^{12}\\
&=1+\theta^{12}+\theta^{18},
\end{aligned}
\]
using \(\theta^{21}=1\) and the fact that the field has characteristic \(2\). Note that
$$
 X^{18}+X^{12}+1 = \left(X^6+X^4+X^2+X+1\right)\cdot\left(X^{12}+X^{10}+X^7+X^6+X^3+X+1\right).
$$
Since $\theta^6+\theta^4+\theta^2+\theta+1=0$,  we obtain
$
1+\theta^{12}+\theta^{18}=0.
$
Thus,
\[
x^3+y^3+z^3+u^3+v^3=0.
\]

\medskip

\noindent
\textbf{Case 3: } $s\ge5$ is odd.
Since $s$ is odd, we have
$
2^s\equiv -1\pmod 3.
$
Then
$
3\mid 2^s+1,
$
which means that
$n$
is divisible by $2^s-1$. Thus,
$
\mathbb{F}_{2^s}^*
\subseteq
\{a\in\mathbb{F}_{2^{2s}}^*:a^n=1\}.
$

Let $\alpha$ be a primitive element of $\mathbb{F}_{2^s}$. From Lemma \ref{eq:5pri}, we know that \(\mathcal C_{(2,2^s-1,5,1)}\) is a binary primitive narrow-sense BCH code
with zeros
$\alpha,\alpha^2,\alpha^3,\alpha^4,$
whose minimum
Hamming distance is $5$. Hence, it contains a codeword of Hamming
weight $5$. Let the support of such a codeword be
\[
\{i_1,i_2,i_3,i_4,i_5\}\subset \mathbb Z_{2^s-1},
\]
where the \(i_j\) are pairwise distinct modulo \(2^s-1\). From the definition of \(\mathcal C_{(2,2^s-1,5,1)}\),  we have
\begin{equation}\label{eq2ss}
\alpha^{i_1}+\alpha^{i_2}+\alpha^{i_3}+\alpha^{i_4}+\alpha^{i_5}=0\,\, \text{and}\,\,\alpha^{3i_1}+\alpha^{3i_2}+\alpha^{3i_3}+\alpha^{3i_4}+\alpha^{3i_5}=0.
\end{equation}

Set
$x=\alpha^{i_1}$, $y=\alpha^{i_2}$, $z=\alpha^{i_3}$, $u=\alpha^{i_4}$ and $v=\alpha^{i_5}$.
The elements \(x,y,z,u,v\) are pairwise distinct because the exponents are distinct modulo \(2^s-1\). Then from (\ref{eq2ss}) we obtain
$$
x+y+z+u+v=0\,\, \text{and}\,\,
x^3+y^3+z^3+u^3+v^3=0.
$$
Finally, each of these elements belongs to \(\mathbb F_{2^s}^*\), whose order divides \(N\). Therefore
\[
x^n=y^n=z^n=u^n=v^n=1.
\]

In all cases we have constructed five pairwise distinct elements satisfying the required conditions. This completes the proof.
\end{proof}

With an argument similar to that used in Theorem \ref{Theoremc1}, we have the following result.

\begin{theorem}\label{Theoremc3}
Let $n=\frac{4^s-1}{3}.$ Then $\mathcal{C}_{(2,n,5,1)}$ has parameters $[n,n-4s,5]$ for $s\geq4$.
\end{theorem}

\begin{remark}
In \cite{Chen59}, the equality
$
d\bigl(\mathcal{C}_{(2,n,5,1)}\bigr)=5
$
was established for even $s$, while the odd case remained open in Conjecture 6.9. When $s=2$ and $s=3$, it is easy to check that $\mathcal{C}_{(2,n,5,1)}$ has parameters $[5,1,5]$ and $[21,12,5]$, respectively.  Hence, Theorem~\ref{Theoremc3} provides a unified argument that covers all possible cases of $s$, and completely confirms Conjecture 6.9 in \cite{Chen59}.
\end{remark}

\begin{example}
Let \( s = 4 \). Then \( \C_{(2,n,5,1)} \) has parameters \([85,69,5]\), which is the optimal binary cyclic code \cite[Appendix A]{DingP1}.
\end{example}

\section{Open Problem on Binary BCH Codes
}

In \cite{Chen59}, the authors posed several open problems concerning three families of binary BCH codes. Open Problem 8.4 asks whether there exist integers \(\delta\) and \(b\) such that the BCH code \(\mathcal C_{(2,n,\delta,b)}\) simultaneously satisfies
\[
\dim (\mathcal C_{(2,n,\delta,b)})\ge \frac{n-1}{2}
\quad\text{and}\quad
d(\mathcal C_{(2,n,\delta,b)})\ge \frac{\sqrt n}{2},
\]
where $n=(2^{2s}+1)(2^s-1)$, or $n=2^{2s}+2^s+1$, or $n=\frac{2^s-1}{\lambda}$ for $\lambda>1$ being a constant divisor of $2^s-1$.
In this section, our main goal is to solve this open problem. We always choose $b=0$ and first give a general construction.

\begin{proposition}\label{pros}
Let $n$ be a positive integer and $m=\ord_n(2)$. Then
 $\dim(\mathcal C_{(2,n,\delta,0)})\ge\frac{n-1}{2}$ if $m\delta\le(n-1)$.
\end{proposition}

\begin{proof}
Let $T$ be the defining set of $\mathcal C_{(2,n,\delta,0)}$. By the
definition of the BCH code,
\[
T=\bigcup_{j=0}^{\delta-2}C_j,
\]
where $C_j$ denotes the binary cyclotomic coset modulo $n$
containing $j$.

For every integer $j$ with $1\leq j\leq\delta-2$, write
$
j=2^a u,
$
where $a\geq0$ and $u$ is odd. Since $1\leq u\leq j\leq
\delta-2$ and $C_{2i}=C_i$, we have
$
C_j=C_u.
$
Consequently,
\[
T=C_0\cup
\bigcup_{\substack{1\leq u\leq\delta-2\\u\ \mathrm{odd}}}C_u.
\]
The number of odd integers in the interval
$[1,\delta-2]$ is exactly
$
\left\lfloor\frac{\delta-1}{2}\right\rfloor.
$
Moreover, $|C_0|=1$, and every nonzero binary cyclotomic
coset modulo $n$ has cardinality at most $m$. Therefore,
\[
|T|\leq
1+m\left\lfloor\frac{\delta-1}{2}\right\rfloor.
\]
Since the dimension of a cyclic code equals its length
minus the cardinality of its defining set, it follows that
\[
\dim (\mathcal C_{(2,n,\delta,0)})
=n-|T|
\geq n-1-m\left\lfloor\frac{\delta-1}{2}\right\rfloor.
\]

If $m\delta\leq n-1$, then
\[
\begin{aligned}
\dim (\mathcal C_{(2,n,\delta,0)})
\geq n-1-m\left\lfloor\frac{\delta-1}{2}\right\rfloor>n-1-\frac{m\delta}{2}\geq n-1-\frac{n-1}{2}=\frac{n-1}{2}.
\end{aligned}
\]
This completes the proof.
\end{proof}

\begin{theorem}\label{eq:thmee}
Let $n=(2^{2s}+1)(2^s-1)$, or $n=2^{2s}+2^s+1$. Let $\delta=\left\lceil\frac{\sqrt n}{2}\right\rceil$,
then \[
\dim (\mathcal C_{(2,n,\delta,0)})\ge \frac{n-1}{2}
\,\,\text{and}\,\,
d(\mathcal C_{(2,n,\delta,0)})\ge \frac{\sqrt n}{2}.
\]
\end{theorem}
\begin{proof}
Let \(T\) be the defining set of \(\mathcal C_{(2,n,\delta,0)}\). By definition,
\[
\{0,1,\dots,\delta-2\}\subseteq T.
\]
Hence, by the BCH bound,
\[
d\bigl(\mathcal C_{(2,n,\delta,0)}\bigr)\ge \delta
=\left\lceil \frac{\sqrt n}{2}\right\rceil
\ge \frac{\sqrt n}{2}.
\]
It remains to prove the dimension bound.
Let \(m=\operatorname{ord}_n(2)\). By Proposition \ref{pros}, it suffices to verify that
$
m\delta\le n-1.
$
We distinguish the two possible forms of \(n\).

\noindent
\textbf{Case 1:} \(n=(2^{2s}+1)(2^s-1)\).
Since
\[
n=(2^{2s}+1)(2^s-1)\mid 2^{4s}-1,
\]
we have \(m\mid 4s\), then \(m\le 4s\).
For \(s=1\), we have \(n=5\) and \(\delta=2\). The defining set is \(T=\{0\}\), so
\[
\dim\bigl(\mathcal C_{(2,5,2,0)}\bigr)=5-1=4\ge \frac{5-1}{2}.
\]
Thus the assertion holds for \(s=1\).
For \(s=2\), we have \(n=51\), \(\delta=4\), and \(m=\operatorname{ord}_{51}(2)=8\). Therefore,
$
m\delta=8\cdot 4=32\le 50=n-1.
$
So Proposition \ref{pros} applies.
Now assume \(s\ge 3\). Using the elementary bound
$
\left\lceil \frac{\sqrt n}{2}\right\rceil\le \frac{\sqrt n}{2}+1,
$
it is enough to prove
\[
4s\left(\frac{\sqrt n}{2}+1\right)<n-1.
\]
Since \(n>2^{3s-1}\) and \(2^{3s-1}\ge (2s+3)^2\) for \(s\ge 3\), we obtain
$
\sqrt n>2s+3.
$
Consequently,
\[
\begin{aligned}
n-1-4s\left(\frac{\sqrt n}{2}+1\right)
&=n-1-2s\sqrt n-4s=\sqrt n(\sqrt n-2s)-4s-1\\
&>(2s+3)\cdot 3-4s-1=2s+8>0.
\end{aligned}
\]
Thus,
$
4s\left(\frac{\sqrt n}{2}+1\right)<n-1.
$
Therefore,
$$
m\delta\le 4s\delta
\le 4s\left(\frac{\sqrt n}{2}+1\right)
<n-1.
$$

\noindent
\textbf{Case 2:}  \(n=2^{2s}+2^s+1\).
Since
\[
n=2^{2s}+2^s+1\mid 2^{3s}-1,
\]
we have \(m\mid 3s\), hence \(m\le 3s\).
For \(s=1\), we have \(n=7\), \(\delta=2\), and \(m=\operatorname{ord}_7(2)=3\). Thus
\[
m\delta=3\cdot 2=6=n-1.
\]
For \(s=2\), we have \(n=21\), \(\delta=3\), and \(m=\operatorname{ord}_{21}(2)=6\). Hence
\[
m\delta=6\cdot 3=18\le 20=n-1.
\]
In both cases Proposition \ref{pros} applies.
Now assume \(s\ge 3\). Again using
$
\left\lceil \frac{\sqrt n}{2}\right\rceil\le \frac{\sqrt n}{2}+1,
$
it suffices to show
\[
3s\left(\frac{\sqrt n}{2}+1\right)<n-1.
\]
Since \(n>2^{2s}\), we have \(\sqrt n>2^s\). For \(s\ge 3\) and \(2^s\ge \frac{3s}{2}+2\), we obtain
$
\sqrt n>\frac{3s}{2}+2.
$
Then
\[
\begin{aligned}
n-1-3s\left(\frac{\sqrt n}{2}+1\right)
&=n-1-\frac{3s}{2}\sqrt n-3s=\sqrt n\left(\sqrt n-\frac{3s}{2}\right)-3s-1\\
&>\left(\frac{3s}{2}+2\right)\cdot 2-3s-1=3>0.
\end{aligned}
\]
Hence,
\[
3s\left(\frac{\sqrt n}{2}+1\right)<n-1.
\]
Therefore,
\[
m\delta\le 3s\delta
\le 3s\left(\frac{\sqrt n}{2}+1\right)
<n-1.
\]

In both cases we have \(m\delta\le n-1\). Applying Proposition \ref{pros} yields
\[
\dim\bigl(\mathcal C_{(2,n,\delta,0)}\bigr)\ge \frac{n-1}{2}.
\]
Together with the distance bound obtained from the BCH bound, the proof is complete.
\end{proof}

When $n=\frac{2^s-1}{\lambda}$, where $\lambda>1$ is a constant divisor of $2^s-1$. In some cases, there do not exist integers \(\delta\) and \(b\) such that the BCH code \(\mathcal C_{(2,n,\delta,b)}\) simultaneously satisfies
\[
\dim (\mathcal C_{(2,n,\delta,b)})\ge \frac{n-1}{2}
\quad\text{and}\quad
d(\mathcal C_{(2,n,\delta,b)})\ge \frac{\sqrt n}{2}.
\]
The following example gives a counterexample.

\begin{example} For the third family $n=\frac{2^s-1}{\lambda}$, the assertion in Open Problem 8.4 \cite{Chen59} does not hold for all admissible pairs $(s,\lambda)$. In particular, for $(s,\lambda)=(18,13797),$ we have $n=19$. By Magma, there exist no integers $\delta$ and $b$ such that $\dim\bigl(\mathcal{C}_{(2,19,\delta,b)}\bigr)\ge 9$ and $ d\bigl(\mathcal{C}_{(2,19,\delta,b)}\bigr)\ge \frac{\sqrt{19}}{2}. $ \end{example}

Although the above example shows that the assertion in Open
Problem~8.4 does not hold for all admissible pairs \((s,\lambda)\),
an affirmative answer can still be obtained under a suitable
condition on \(n\) and \(s\).

\begin{theorem}
Let $n\geq 5$, $4n\geq s^2$ and $n=\frac{2^s-1}{\lambda}$ for $\lambda>1$ being a constant divisor of $2^s-1$. Let $\delta=\lceil\frac{\sqrt n}{2}\rceil$,
then \[
\dim (\mathcal C_{(2,n,\delta,0)})\ge \frac{n-1}{2}
\quad\text{and}\quad
d(\mathcal C_{(2,n,\delta,0)})\ge \frac{\sqrt n}{2}.
\]
\end{theorem}

\begin{proof}
Let $T$ be the defining set of $\C_{(2,n,\delta,0)}$ and $m=\ord_n(2).$ Since \(C_{2i}=C_i\), every nonzero cyclotomic coset appearing in
\(T\) can be represented by an odd integer in
\(\{1,2,\ldots,\delta-2\}\). Therefore,
\[
|T|
\le
1+m\left\lfloor\frac{\delta-1}{2}\right\rfloor.
\]
Consequently,
\[
\dim (\mathcal C_{(2,n,\delta,0)})
=n-|T|
\ge
n-1-m\left\lfloor\frac{\delta-1}{2}\right\rfloor.
\]

Since
$
4n\ge s^2,
$
we have
$
s\le2\sqrt n.
$
Moreover, since
$
\delta=\left\lceil\frac{\sqrt n}{2}\right\rceil,
$
we obtain
$
\delta-1<\frac{\sqrt n}{2},
$
and hence
$
\left\lfloor\frac{\delta-1}{2}\right\rfloor
<
\frac{\sqrt n}{4}.
$
Using \(m\le s\), we obtain
\[
\begin{aligned}
2m\left\lfloor\frac{\delta-1}{2}\right\rfloor
\le
2s\left\lfloor\frac{\delta-1}{2}\right\rfloor
<
\frac{s\sqrt n}{2}
\le n.
\end{aligned}
\]
Since
$
2m\left\lfloor\frac{\delta-1}{2}\right\rfloor
$
is an integer, it follows that
$
2m\left\lfloor\frac{\delta-1}{2}\right\rfloor
\le n-1.
$
Therefore,
\[
m\left\lfloor\frac{\delta-1}{2}\right\rfloor
\le\frac{n-1}{2},
\]
and consequently,
\[
\begin{aligned}
\dim (\mathcal C_{(2,n,\delta,0)})
\ge
n-1-
m\left\lfloor\frac{\delta-1}{2}\right\rfloor\ge
\frac{n-1}{2}.
\end{aligned}
\]

On the other hand, by the BCH bound,
$
d\bigl(\mathcal C_{(2,n,\delta,0)}\bigr)
\ge\delta
\ge\frac{\sqrt n}{2}.
$

This completes the proof.
\end{proof}

\section{Conclusion}

In this paper, we investigated several conjectures and an open problem concerning the parameters of three families of binary BCH codes. In addition to settling the proposed conjectures, we further studied a BCH code whose minimum distance problem is more involved than those appearing in the original conjectures.
The main contributions of this paper are as follows:

\begin{itemize}

\item We completely settled the three conjectures proposed in \cite{Chen59} on the exact minimum distances of three families of binary BCH codes. In each case, the BCH lower bound is shown to be attained by explicitly constructing a codeword of the corresponding minimum weight.

\item Beyond the original conjectures, we further investigated the more difficult code
\[
\mathcal{C}_{(2,2^{2s}+2^s+1,5,1)}.
\]
For several infinite classes of \(s\), we proved that its minimum distance is \(5\).

\item We studied Open Problem~8.4. Affirmative answers were obtained for the first two length families. For the third family
$
n=\frac{2^s-1}{\lambda},
$
we established a sufficient condition for the desired BCH codes to exist and provided a counterexample showing that the unrestricted statement is false in general.

\end{itemize}

These results settle all the conjectures considered in \cite{Chen59}, extend the study to a more difficult minimum-distance problem, and clarify the scope of Open Problem~8.4.

\end{document}